\documentclass[11pt,a4paper,oneside]{article} 

\usepackage{mathptmx}
\usepackage{graphics,epsfig,latexsym,amssymb,amsbsy,amsmath,amsthm}

\usepackage[numbers]{natbib}

\usepackage{graphics}
\usepackage{geometry}

\usepackage{xcolor}
\usepackage{graphicx}

\usepackage{titlesec}

\titleformat*{\subsection}{\bfseries}
\titleformat*{\subsubsection}{\bfseries}

\theoremstyle{plain}
\newtheorem{theorem}{Theorem}[section]
\newtheorem{lemma}[theorem]{Lemma}

\newcommand{\rP}{\mbox{P}}
\newcommand{\THWD}{T_{\text{HWD}}}
\newcommand{\WHWD}{W_{\text{HWD}}}
\newcommand{\TCATT}{T_{\text{CATT}}}
\newcommand{\W}{W_{\hat\pi}}
\newcommand{\Wpi}{W_\pi}
\newcommand{\Wd}{W_{\delta}}
\newcommand{\Wde}[1]{W_{\delta=#1}}
\newcommand{\T}{T}
\newcommand{\V}{T_{\chi^2}}

\newcommand{\hqpione}{\hat{q}_{1,\hat\pi}}
\newcommand{\hqpitwo}{\hat{q}_{2,\hat\pi}}
\newcommand{\hqpone}{\hat{q}_{1,p}}
\newcommand{\hqptwo}{\hat{q}_{2,p}}
\newcommand{\hqhpii}{\hat{q}_{i,\hat\pi}}
\newcommand{\hqone}{\hat{q}_{1}}
\newcommand{\hqtwo}{\hat{q}_{2}}
\newcommand{\hqi}{\hat{q}_{i}}
\newcommand{\hqhpioneone}{\hat{q}_{11,\hat\pi}}

\newcommand{\Zcatt}{T_{\text{CATT}}}

\usepackage[normalem]{ulem}                         
\usepackage{color}                                  

\definecolor{darkgreen}{rgb}{0,0.5,0}

\begin{document}

\title{A powerful MAF-neutral allele-based test\\for case-control association studies}

\author{%
M.~A. Jonker\textsuperscript{a}, J. Pecanka\textsuperscript{b}\\[2em]
\small
\textsuperscript{a}Radboud Institute for Health Sciences, Radboud University Medical Center, Nijmegen, Netherlands\\
\small
\textsuperscript{b}Department of Biomedical Data Sciences, Leiden University Medical Center, Leiden, Netherlands\\[3em]
\small
Corresponding author: M.~A. Jonker, marianne.jonker@radboudumc.nl
}
\date{December 1, 2018}

\maketitle

\begin{abstract}
{In a case-control study aimed at locating autosomal disease variants for a disease of interest, association between markers and the disease status is often tested by comparing the marker minor allele frequencies (MAFs) between cases and controls. For most common allele-based tests the statistical power is highly dependent on the actual values of these MAFs, where associated markers with low MAFs have less power to be detected compared to associated markers with high MAFs. Therefore, the popular strategy of selecting markers for follow-up studies based primarily on their p-values is likely to preferentially select markers with high MAFs.
We propose a new test which does not favor markers with high MAFs and improves the power for markers with low to moderate MAFs without sacrificing performance for markers with high MAFs and is therefore superior to most existing tests in this regard. An explicit formula for the asymptotic power function of the proposed test is derived theoretically, which allows for fast and easy computation of the corresponding p-values. The performance of the proposed test is compared with several existing tests both in the asymptotic and the finite sample size settings.}
\end{abstract}


\newpage

\section{Introduction}

When locating dichotomous trait loci (such as disease variants) at autosomal chromosomes, association studies of genetic markers are typically conducted using the case-control study design. Over the years, a fair number of genetic association tests suitable for such studies have been proposed \cite{Balding06,Zheng12}. For autosomal markers the native test would be based on genotypic information, however, tests contrasting the observed marker allele frequencies in the samples of cases and controls are often preferentially used due to their beneficial properties such as an ability to reliably recover signals even under deviations from additivity of allelic effects (e.g.\ under a dominance or recessive model). Among the existing tests of this type, probably the best known example is the binomial test of equality of allele frequencies in the samples of cases and controls, henceforth called the \emph{allele-based test} (ABT). Other popular alternatives are the chi-square test for association, the Fisher exact test, the logistic regression model (LRM) score test, and the Cochran-Armitage trend test (CATT) \cite{Balding06,Sasieni97,Zheng12}. The last of these has the advantage of being applicable even when the assumption of Hardy-Weinberg equilibrium is violated, while the score test stands out due to its abilities to adjust for potential confounders and to model multiple markers (including interactions) simultaneously.

By definition, a marker is associated with a disease, or more generally with a dichotomous trait, if it is in linkage disequilibrium (LD) with one of its causal genetic variants \cite{Zheng12}. For most existing tests, including those mentioned above, the power to detect a marker is highly dependent on the degree of LD between the marker and the causal variant. Typically, the stronger the LD the smaller the p-value of the test. However, the p-values also depend on the marker allele frequencies; among markers that are in LD with the same causal variant, markers with high minor allele frequencies (MAFs) are typically much more likely to be detected than markers with low MAFs. Consequently, the strategy of selecting individual markers for follow-up studies primarily using the p-values from the existing tests is biased towards selecting markers with high MAFs. The same holds for most alternative strategies for prioritizing markers for follow-up studies that have been proposed in the literature such as ranking markers using the Bayes factor \cite{Wakefield08,Wakefield09}, the likelihood ratio signal \cite{Stromberg09}, the frequentist factor \cite{Wacholder04}, or PrPES \cite{Stromberg08} as the signal measures. A comparison between these strategies and the strategy of ranking markers using the p-values of various allele-based tests and the CATT found that all of the considered strategies resulted in highly similar ordering of markers and the markers with the smallest p-values obtained from the ABT tended to be top-ranked by the other methods as well, and vice-versa \cite{Stromberg09}. In fact, some of the alternative strategies exhibited a tendency to disfavor markers with small MAFs to an even stronger degree than the ABT p-value based ranking.

In this paper we propose a novel test which can be viewed as an adjustment of the standard ABT for testing association in case-control studies, which reduces the preferential treatment of markers with high MAFs. We show that the new test has equivalent or superior power compared to the commonly used tests, and the power superiority occurs particularly in situations with low to moderate marker MAFs. We also show how the new test can be made robust against deviations from Hardy-Weinberg equilibrium, an important practical concern. We derive an explicit formula for the test's asymptotic power function, thus allowing for fast and easy computation of the test's p-values. A comparison is made with the (asymptotic) power function for the standard ABT, the CATT, the chi-square test and the LRM score test in the absence of confounders. In addition to the asymptotic perspective, we also investigate via simulation the power performance of the new test in a finite sample size setting. Finally, we apply the new test to a major depression disorder case-control data set.

\section{Methods}

\subsection{Setting}

In this paper we define \textsl{causal variant}, or simply \textsl{variant}, to mean a causal genetic locus (e.g.\ SNP) and \textsl{marker} to mean an observed genetic locus, which may or may not be in LD with a causal variant. For the disease of interest there may be multiple causal variants. The goal is to identify the markers that are in LD with any of the causal variants for the given disease of interest. For simplicity of notation, we assume that there is only one causal variant. In Section \ref{sec.Multiple_variants} we briefly discuss the situation with multiple causal variants.

The case-control status with respect to a given disease of interest for a random individual from a specific population is denoted by $A$ whenever the individual is a case (i.e.\ is affected by the disease, thus also called \textsl{unaffected}) and by $U$ for a control (also called \textsl{unaffected}). Furthermore, the fraction of the cases in the total population is denoted by $\pi$.

Suppose the causal variant is biallelic with alleles $A_1$ and $A_2$. Denote the corresponding allele frequencies in the total population as $p_1$ and $p_2=1-p_1$ and the frequencies of $A_i$ only among the controls (unaffected) and cases (affected) as $p_i^U$ and $p_i^A$, respectively. Note that, trivially, it holds $p_1^U+p_2^U=p_1^A+p_2^A=1$. If a variant exhibits more than two alleles, it can still be treated as biallelic by re-defining one of the alleles, say $A_2$, to denote any allele that is not $A_1$. Further denote the fraction of the cases among the individuals with genotype $(A_i,A_j)$ at the causal variant by $\pi_{ij}$. Since genotypes are non-ordered, it is assumed that $\pi_{12}=\pi_{21}$. Further it is assumed that all markers are also biallelic. For a given marker, the two alleles are denoted by $M_1$ and $M_2$ with $q_1$ and $q_2$ the corresponding frequencies in the total population. Similarly denote the frequencies of $M_1$ and $M_2$ only among the controls (``unaffected'') as $q_1^U$ and $q_2^U$ and only among the cases (``affected'') as $q_1^A$ and $q_2^A$, respectively. Note that, again trivially, it holds $q_1^U+q_2^U=q_1^A+q_2^A=1$.

In the following sections until Section \ref{sec.AdjustmentsHWD} it is assumed that the marker alleles are in Hardy-Weinberg equilibrium (HWE). In Section \ref{sec.AdjustmentsHWD} we present adjustments of our test (defined in (\ref{eqn.hqhpii}) below) aimed at situations where the assumption of HWE is violated. Additionally, throughout the paper it is assumed that genotyping errors can be neglected and that the samples of cases and controls are random and independent selections from the cases and controls in the given population of interest.

The total sample consists of $N$ independent individuals of which there are $R$ cases and $S$ controls, where $R$ and $S$ are assumed to be fixed and non-random. In other words, for biallelic markers we observe a total of $2R$ and $2S$ alleles for the cases and the controls, respectively. Let $R_0$ and $R_1$ denote the observed counts of genotypes $(M_1,M_1)$ and $(M_1,M_2)$ among the cases, respectively, and let $S_0$ and $S_1$ denote the corresponding genotype counts among the controls. We then estimate the frequencies of the allele $M_1$ among the cases and the control by $\hqone^A=\tfrac{1}{2}(2R_0+R_1)/R$ and $\hqone^U=\tfrac{1}{2}(2S_0+S_1)/S$, respectively. We denote the estimates of the complementary frequencies as $\hqtwo^A=1-\hqone^A$ and $\hqtwo^U=1-\hqone^U$.

\subsection{Novel test statistic}

Inspired by the binomial test of equality of allele frequencies of cases and controls (ABT) for testing H$_0:q_1^U=q_1^A$ versus H$_1:q_1^U\neq{}q_1^A$, we propose a novel test statistic for testing H$_0$ against H$_1$. We define the statistic as
\begin{align}
\W\;=\;\frac{\sqrt{m}(\hqone^U - \hqone^A)}{\sqrt{\hqpione\hqpitwo}},
\label{eqn.Wdelta}
\end{align}
where $m=2N\lambda(1-\lambda)$ with $\lambda=R/N$ (i.e. the fraction of the cases in the sample) and
\begin{align}
\label{eqn.hqhpii}
\hqhpii\;=\;\hat\pi\hqi^A+(1-\hat\pi)\hqi^U,\qquad i=1,2,
\end{align}
with $\hat\pi$ denoting an estimate of the disease prevalence $\pi$. This latter estimate cannot be obtained from the case-control data and thus additional external information is required for the estimation. For many diseases suitable estimates of the population prevalence are readily available from literature or other sources such as national registries (see also Section \ref{sec.Robustness}).

Assuming that $\hat\pi$ is an asymptotically consistent estimator of $\pi$, the denominator in $\W$ converges in probability to $\sqrt{q_1q_2}$ as the number of observations in the case-control sample and used for estimating $\pi$ increase to infinity. Consequently, by Slutsky's lemma and the central limit theorem, it follows that $\W$ is asymptotically standard normally distributed under the null hypothesis H$_0$. In other words, rejecting H$_0$ whenever $|W_{\hat\pi}|>\xi_{\alpha/2}$, where $\xi_{\alpha/2}$ is the upper $\alpha/2$-quantile of the standard normal distribution, yields a test of H$_0$ against H$_1$ with an asymptotic level of significance of $\alpha$.

\subsubsection*{Motivation}

The motivation for the new statistics comes from the following equality, which we derive in Appendix A,
that reads
\begin{align}
\frac{q_{1}^U-q_{1}^A}{\sqrt{q_1q_2}}\;=\;\Delta\;\frac{p_{1}^U-p_{1}^A}{\sqrt{p_1p_2}},
\label{eqn.pMLUDelta}
\end{align}
where $\Delta$ is a common measure for the degree of LD between a marker and a causal variant, defined as $\Delta := D_{11}/\sqrt{p_1p_2q_1q_2}$, where $D_{ij}=\rP(A_iM_j)-p_iq_j$ and $\rP(A_iM_j)$ denotes the frequency of the joint haplotype $(A_i,M_j)$ at the causal variant and the marker in the total population \cite{Devlin95}. The equality (\ref{eqn.pMLUDelta}) shows how the relative difference between the allele frequencies among the cases and the controls at the causal variant (the quotient on the right-hand side of (\ref{eqn.pMLUDelta})) is passed on to the neighboring markers through the multiplication by $\Delta$. An immediate consequence of (\ref{eqn.pMLUDelta}) is this. If the marker allele frequencies among the controls and among the cases are unequal ($q_1^U-q_1^A\neq 0$), then $\Delta$ must be non-zero, and vice versa \cite{Kruglyak99,Pritchard01}. Since the goal of an association analysis is to find markers for which $\Delta\neq0$, it follows that testing the null hypothesis H$_0:\Delta=0$ against the alternative hypothesis H$_1:\Delta\neq0$ is equivalent to testing H$_0:q_1^U=q_1^A$ against H$_1:q_1^U\neq{}q_1^A$. Since typically only marker data is available, the equation (\ref{eqn.pMLUDelta}) naturally suggests to use a test statistic that is of the form of the left-hand side of (\ref{eqn.pMLUDelta}). Hence the new statistic $\W$.

\subsection{Asymptotic power functions: A comparison}

In this section we present an (asymptotic) power comparison of $\W$ and several commonly used tests of equality of allele frequencies as well as the classical chi-square test statistic denoted as $\V$ \cite{Sasieni97}. A commonly used frequency-based tests utilize the statistic $\T$ defined as
\begin{align*}
\T&\;=\;\frac{\sqrt{m}(\hqone^U-\hqone^A)}
{\sqrt{\lambda\hqone^U\hqtwo^U+(1-\lambda)\hqone^A\hqtwo^A}}.
\end{align*}
Under the null hypothesis of no association $\T$ is asymptotically standard normally distributed.

In addition to $\T$, two other tests of association are popularly used. Namely the Cochran-Armitage trend test \cite{Zheng12}, for which we denote the statistic by $\Zcatt{}$, and the LRM score test where the observed minor allele count is the independent variable \cite{Balding06}. Their powers are compared with that of $\W$ using a theoretical argument.

For the sake of brevity, in this paper we only focus on the additive model. However, our investigation (not shown) has indicated that the presented conclusions remain qualitatively true for other genetic models including the dominant and the recessive models as well as other parameter settings.

\subsubsection*{Power comparison between $\W$ and $\T$: theory}

It is easy to see that the statistics $\T$ and $\W$ are closely linked. It is straightforward to show that $\W=\hat{Q}_{\hat\pi}^{-1}\T$, where
\begin{align*}
\hat Q_{\hat\pi}\;=\;\frac{\sqrt{\hqpione\hqpitwo}}
{\sqrt{\lambda\hqone^U\hqtwo^U+(1-\lambda)\hqone^A\hqtwo^A}}.
\end{align*}
Assuming that $\hat\pi$ is asymptotically consistent and that the fraction of the cases $\lambda$ is fixed, $\hat{Q}_{\hat\pi}$ converges in probability to $Q$, where $Q^2=q_1 q_2/(\lambda q_{1}^U q_{2}^U+(1-\lambda)q_{1}^Aq_{2}^A)$, as $m$ and as all of the sample sizes underlying $\hat\pi$ go to infinity. Under the alternative hypothesis it holds $Q\neq1$, thus $\T$ and $\W$ do not have equal power. However, they do have the same level since under the null hypothesis it holds $Q=1$ (since $q_1^U=q_1^A=q_1$). In fact, under the null hypothesis $\hat{Q}_{\hat\pi}$ converges in probability to 1 irrespective of the asymptotic consistency of $\hat\pi$. More specifically, if $\hat\pi$ is replaced in $\hat{Q}_{\hat\pi}$ by any value $\delta\in(0,1)$, the resulting fraction $\hat{Q}_{\delta}$ still converges in probability to 1, meaning that for any $\delta\in(0,1)$ in place of $\hat\pi$ the corresponding test is valid (see Section \ref{sec.Robustness} for further discussion).

Further investigating the link between $\W$ and $\T$, an application of the central limit theorem and Slutsky's lemma yields that for $m$ and the numbers of observations underlying $\hat\pi$ all going to infinity it holds
\begin{align*}
\W\;=\;\frac{\sqrt{m}(\hqone^U - \hqone^A)}{\sqrt{\hqpione\hqpitwo}}
\;=\;\frac{\sqrt{m}(q_{1}^U - q_{1}^A)}{\sqrt{q_1q_2}} + O_p(1)
\;=\;\frac{\sqrt{m} \Delta (p_{1}^U - p_{1}^A)}{\sqrt{p_1p_2}} + O_p(1),
\end{align*}
where $O_p(1)$ denotes a term that is bounded in probability. Note that the last equality follows from (\ref{eqn.pMLUDelta}).
Consequently, with $B=(p_{1}^U - p_{1}^A)/\sqrt{p_1 p_2}$, it holds that
\begin{align}
\W\;=\;\sqrt{m} \Delta B + O_p(1)\qquad \mbox{ and } \qquad
\T\;=\;\sqrt{m} \Delta B Q + O_p(1).
\label{eqn.approx_W_and_T}
\end{align}
In other words, for large $m$ the power functions of the tests based on $\W$ and $\T$ are respectively governed by the terms $\sqrt{m}\Delta{}B$ and $\sqrt{m}\Delta{}B{}Q$. Note that there are three types of quantities at play here. The term $\sqrt{m}$ is sample-specific and is the same for every marker. $\Delta$, on the other hand, expresses the degree of LD between the marker and the causal variant and is therefore marker-specific, and so is the term $Q$. Finally, $B$ is specific to the causal variant and is therefore the same for all markers that are in LD with the same causal variant.

In terms of power, the asymptotic approximations in (\ref{eqn.approx_W_and_T}) show that for each marker the p-values based on $\T$ are weighted by their sample allele frequencies via $Q$, where $Q\neq 1$ under the alternative hypothesis. In the case of $\W$, however, the term $Q$ is absent, which means that there is no frequency-based weighing and thus the corresponding p-values are much more comparable over markers with different allele frequencies, especially if these markers are in LD with the same causal variant and thus have the same value for $B$.

Suppose that the minor allele $M_1$ is positively correlated with the risk allele at the variant. Then, $M_1$ will be enriched among the cases, and thus $q_1^U<q_1<q_1^A$ and $q_2^A<q_2<q_2^U$. If $q_1^A<\tfrac{1}{2}$, then $q_1^Uq_2^U<q_1q_2<q_1^Aq_2^A$, because the function $p \rightarrow p(1-p)$ is concave and symmetric around $\tfrac{1}{2}$.
Recall that $\lambda$ equals the fraction of cases. It holds $Q\leq 1$ if and only if $\lambda\leq \lambda_0$, with $\lambda_0=(q_1^Aq_2^A-q_1q_2)/(q_1^Aq_2^A-q_1^Uq_2^U)$, at which point $\W$ is more powerful than $\T$. In most practical situations $\lambda=\tfrac{1}{2}$ (balanced design) or $\lambda<\tfrac{1}{2}$ (more controls than cases). Although, in practice there might be markers for which $\lambda>\lambda_0$, this will not be common. For $M_1$ negatively correlated with the risk allele at the variant, the power ordering between $\W$ and $\T$ is reversed. However, settings with strong or even mild negative correlations between the minor allele $M_1$ and the minor risk allele at the causal variant are not generally possible.

\subsubsection*{Power comparison of $\W$ with $\T$ and $\V$: numerical results}

In the top row plots in Figure \ref{fig.power_asymptotic} we provide a numerical comparison of $\W$ with $\T$ and $\V$ in terms of their power performances. The asymptotic power functions of $\W$ (continuous line), $\T$ (dashed line) and $\V$ (dotted line) are shown as a function of $q_1$ (left plot) and of $\Delta$ (right plot). The number of cases and controls was put at $R=S=10,000$ and the significance level was set to $\alpha=5\times10^{-8}$. Notice that in both plots the power functions of $\T$ and $\V$ almost completely overlap, which means that the two statistics have almost identical power. Moreover, the power functions of $\T$ and $\V$ lie fully below that of $\W$, which shows $\W$ to be more powerful than both $\T$ and $\V$ in the considered setting. In terms of MAF, $\W$ is the superior performer for a majority of values. Unsurprisingly, the degree of superiority of $\W$ weakens with increasing $q_1$ until the ordering flips for MAF near 0.5, when the statistics $\T$ and $\V$ both become (slightly) more powerful than $\W$.

In the bottom row plots in Figure \ref{fig.power_asymptotic} the power functions for $\W$ (continuous lines) and $\T$ (dashed lines) are given for the same setting as before except here the design is unbalanced with the number of cases and controls set equal to $R=6000$ and $S=16,000$ (left) and to $R=16,000$ and $S=6000$ (right). It shows the power of $\T$ to be dependent on the fraction of the cases in the sample. Clearly, the more unbalanced in favor of the controls the design is the more the corresponding test favors markers with large MAFs.

As discussed in the theoretical part, the power function based on $\W$ is constant as a function of $q_1$, while the power functions of $\T$ and $\V$ increase with $q_1$. These properties drive the behavior of these statistics for various MAFs. It explains why $\T$ and $\V$ both favor markers with large MAFs at the cost of those with smaller MAFs, and why $\W$ does not exhibit such behavior. A direct consequence of these properties is that the p-values based on $\W$ are much more comparable across markers with different MAFs.

Figure \ref{fig.power_robust} (right) shows the power functions for a more prevalent disease. A qualitatively similar behavior has been observed under a number of alternative settings (results not shown).

\begin{figure}
\begin{center}
\includegraphics[scale=0.55]{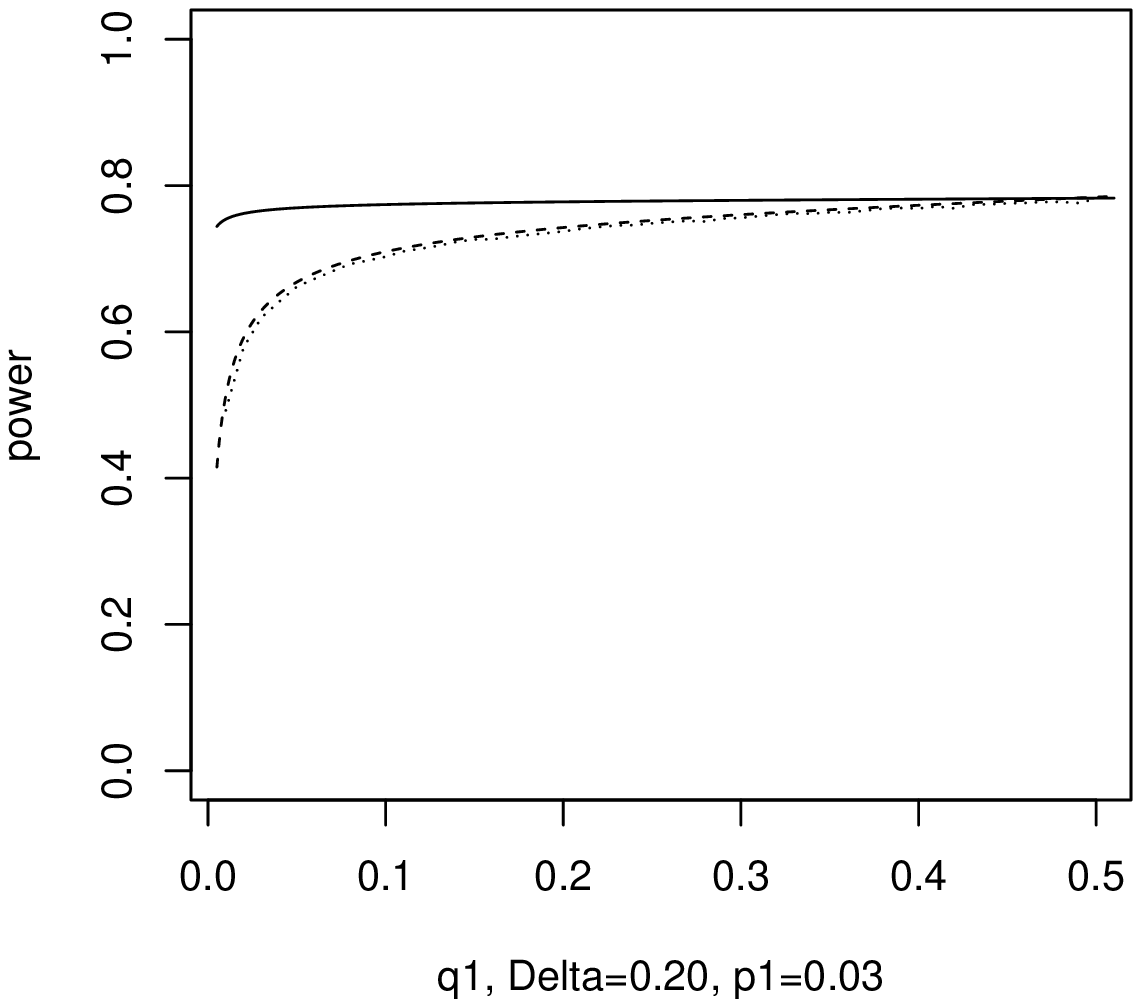}
\includegraphics[scale=0.55]{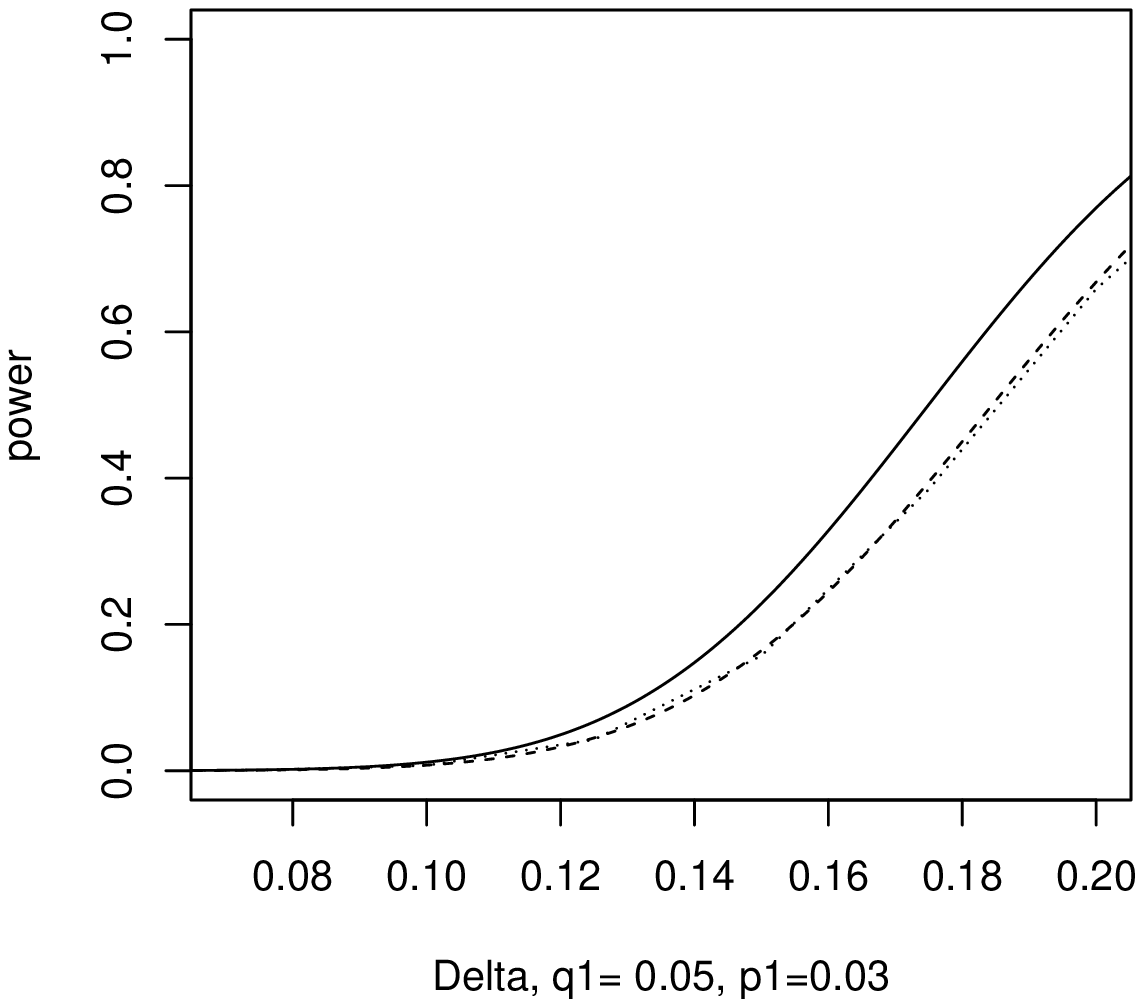}
\includegraphics[scale=0.55]{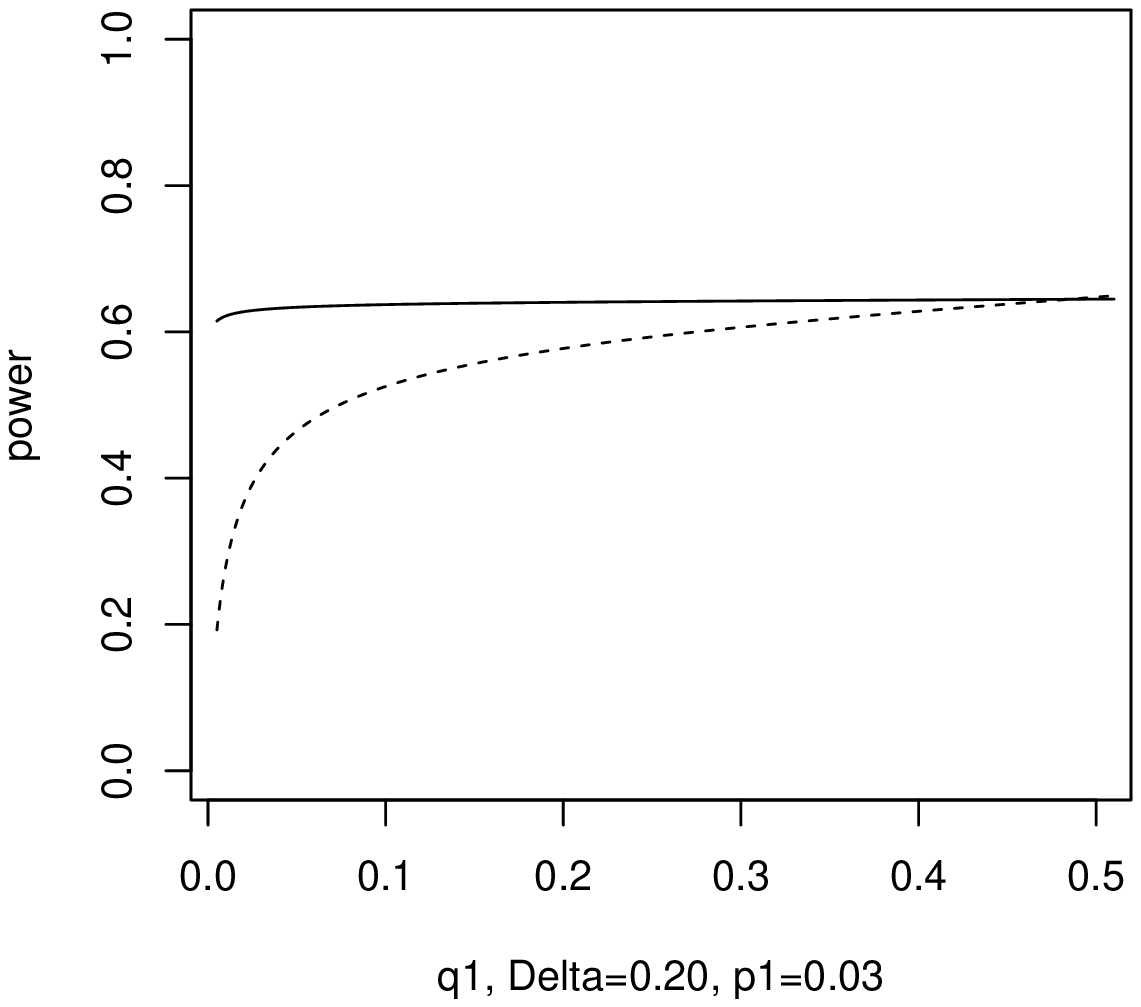}
\includegraphics[scale=0.55]{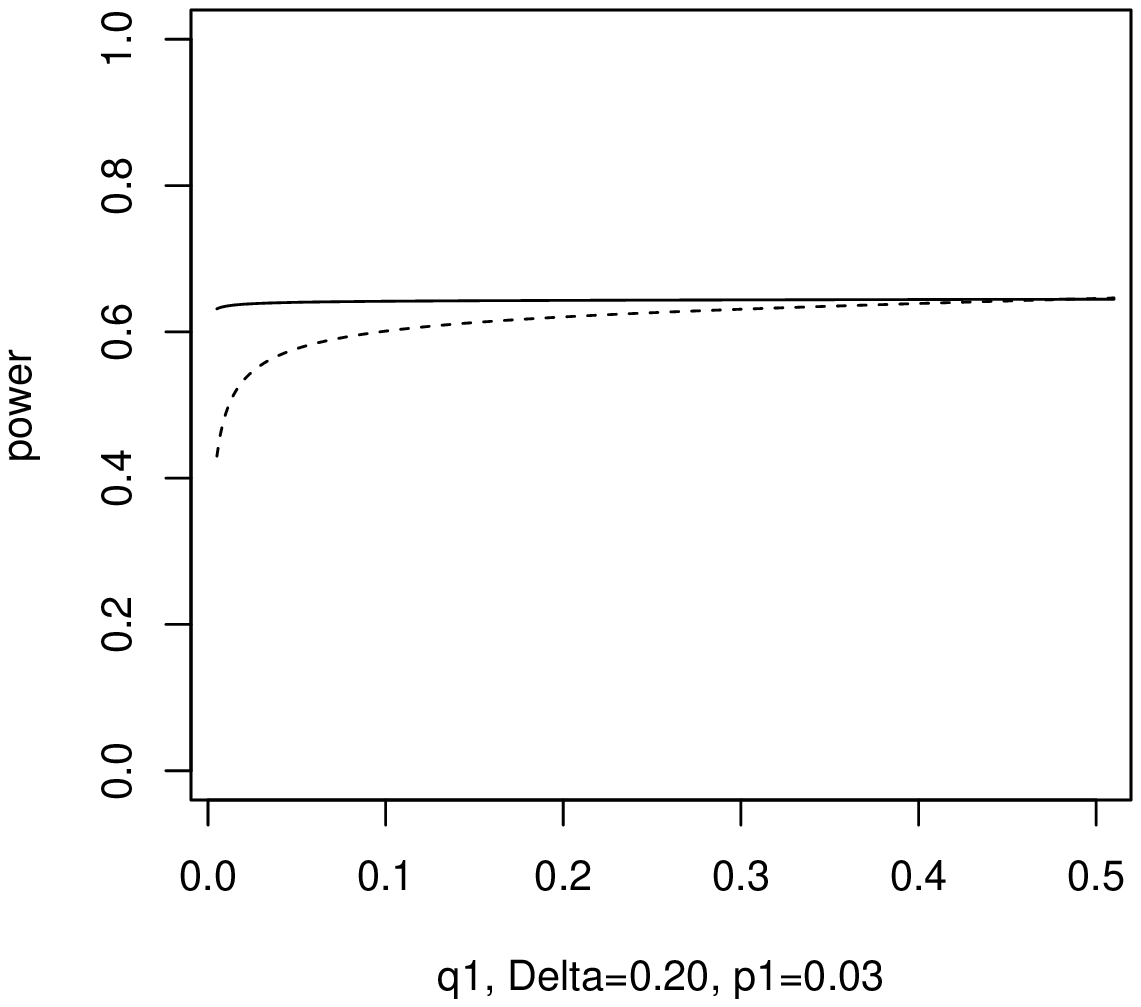}
\end{center}
\caption{Power functions for $\W$ (continuous lines), $\T$ (dashed lines), $\V$ (dotted line). Additive model with $p_1=0.03, \pi_{11}=0.10, \pi_{22}=0.02$, $\pi_{12}=0.06$. Top row: $R=S=10,000$ (balanced design). Bottom row left: $R = 6000$, $S=16,000$ (unbalanced design). Bottom row right: $R=16,000$, $S=6000$ (unbalanced design).}
\label{fig.power_asymptotic}
\end{figure}

\subsubsection*{Power comparison of $\W$ with $\TCATT$ and the LRM score test}

It has been shown that for the additive model and under the assumption of HWE, the LRM score test (with the observed minor allele count as independent variable) is equivalent to the CATT \cite{Zheng12}. Consequently, in this setting any test that is more powerful than the CATT is also more powerful than the score test, and vice-versa. In other words, it is sufficient to compare the powers of the test based on $\W$ and the CATT test.

Under HWE, for the CATT test statistic under the additive model ($\Zcatt{}(1/2)$) it holds
\begin{align*}
T^2 \;=\; \Zcatt{}^2(1/2) \frac{\hqpone\hqptwo}{\lambda \hqone^U\hqtwo^U + (1-\lambda)\hqone^A\hqtwo^A},
\end{align*}
with $\hqpone$ and $\hqptwo$ the pooled sample estimators for the $q_1$ and $q_2$ allele frequencies \cite{Zheng12}. This in turn yields
\begin{align}
\label{eqn.Wpi2}
\W^2\;=\;\Zcatt{}^2(1/2)\,\frac{q_{1,p}q_{2,p}}{q_{1,\pi}q_{2,\pi}}+o_P(1),
\end{align}
with $o_P(1)$ denoting a term that converges in probability to zero. Under the null hypothesis of no association, the fraction term in (\ref{eqn.Wpi2}) equals 1, meaning that the tests based on $\Wpi$ and $\Zcatt{}(1/2)$ have the same asymptotic level of significance. Under the alternative hypothesis, assuming that the minor allele $M_1$ is positively correlated with the causal variant (i.e.\ the sample of cases is enriched with carriers of the risk alleles at the causal variants) and the prevalence of cases is higher in the pooled sample than in the population, $\W$ is more powerful than $\Zcatt{}(1/2)$. This is because then the fraction term in (\ref{eqn.Wpi2}) is expected to exceed one, which leads to $q_{1,p}>q_{1,\pi}$ and $q_{1,p}q_{2,p}>q_{1,\pi}q_{2,\pi}$. Moreover, under this setting $\W$ is also more powerful than the LRM score test.

\subsubsection*{Take-away message of the comparisons}

The theoretical and the numerical results presented in this section show that under HWE and for the additive model, the test based on $\W$ is, under many relevant situations, more powerful than the test based on $\T$, $\Zcatt{}$, $T_{\chi^2}$ and the LRM score test. Moreover, the power functions for $\W$ are constant, indicating that the test does not favor markers with high MAFs, contrary to the other test considered.

\subsection{Robustness of $\W$ against misspecification of $\pi$}
\label{sec.Robustness}

As mentioned, $\W$ relies on an external source for an accurate estimate of the population prevalence $\pi$ which cannot be directly derived from the case-control data at hand. Fortunately, the information on disease prevalence often can be acquired from literature or relevant national registries (e.g.\ disease prevalences in the Netherlands are published by the National Institute for Public Health and the Environment). If no reliable estimate of $\pi$ can be obtained, a reasonable value can be guessed by relevant experts. Nonetheless, even if good estimates are available, it is relevant to study the robustness of the performance of $\W$ with respect to the quality of the estimate of $\pi$. For a fixed $\delta\in(0,1)$ define the test statistic $\Wd$ to be equal to $\W$ evaluated at $\hat\pi=\delta$. Under the null hypothesis, where $q_1^U=q_1^A$, the denominator of $\Wd$ converges in probability to $\sqrt{q_1q_2}$ irrespective of the value of $\delta$. Consequently, the type I error of the test based on $\W$ is insensitive to the quality of the prevalence estimate. However, the power of the test is dependent on the estimate for the prevalence. In Figure \ref{fig.power_robust}, on the left, the asymptotic power functions of $\T$ and $W_{\delta=\pi}$, $W_{\delta=0.05}$, $W_{\delta=0.1}$, $W_{\delta=0.2}$, $W_{\delta=0.3}$ as functions of $q_1$ are plotted. The value of $\pi$ was set at $\pi=0.0224$, while the other parameters were set at $\Delta=0.20, p_1=0.03$. The figure shows that for $\delta$ equal to or near $\pi$ the power functions are more or less constant with respect to the MAF, while for values of $\delta$ far from $\pi$ the power functions do vary with the MAF (they increase). For values of $\delta<\pi$ (underestimation of the prevalence) the power function of the test is slightly above that for $\delta=\pi$, although the difference is small and diminishes with increasing allele frequency $q_1$.

In Figure \ref{fig.power_robust} (right) the asymptotic power functions correspond to a setting of a more common disease, namely $p_1=0.2$, $\pi_{11}=0.40$, $\pi_{12}=0.10$, $\pi_{22}=0.25$, which yields $\pi=0.16$, and $R = S = 4000$. The power curves are for the test-statistic $W_\delta$ with $\delta=0.01$, $\delta=0.05$, $\delta=0.10$, $\delta=\pi$, $\delta=0.20$ (ordered top to bottom) and for $T$ (dotted line). The plot shows a flat power function for $W_{\delta=\pi}$, a slightly decreasing function for $\delta<\pi$ and a slightly increasing function for $\delta>\pi$. It also shows the robustness of the power of $W_\delta$ against minor misspecification of $\pi$ for both overestimated and underestimated $\pi$.

\begin{figure}
\begin{center}
\includegraphics[scale=0.55]{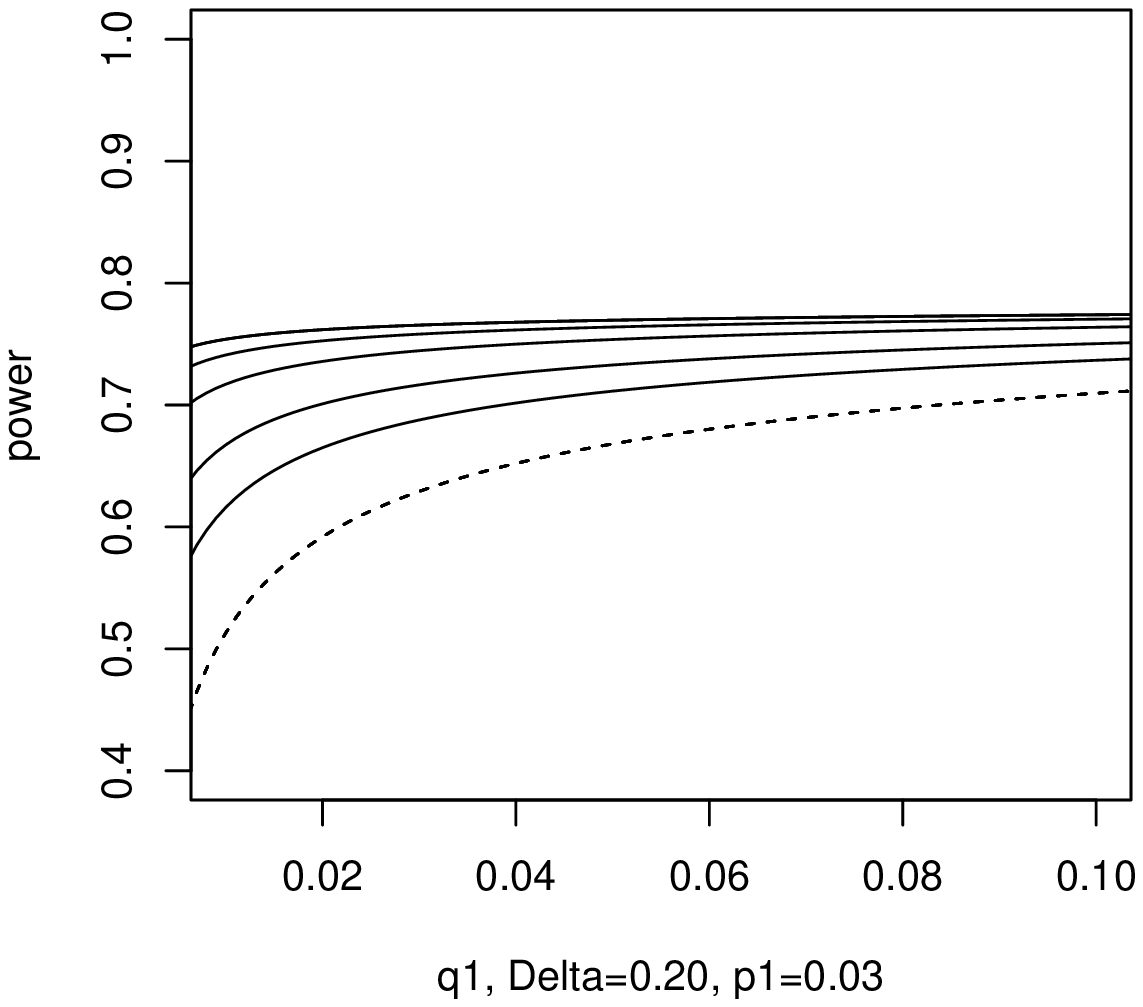}
\includegraphics[scale=0.55]{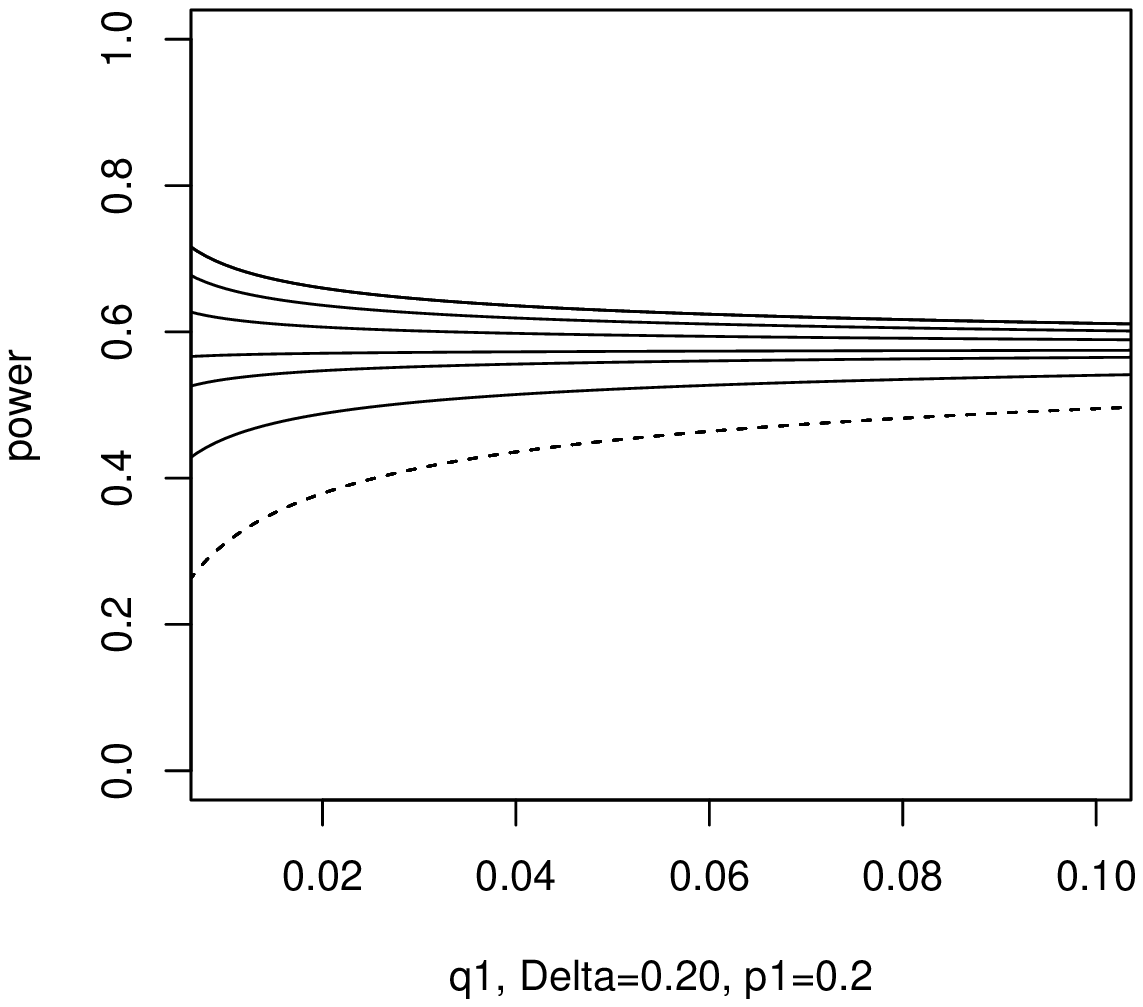}
\end{center}
\caption{Left plot: Power functions of $\T$ (dotted line) and $W_{\delta=0.01}$, $W_{\delta=\pi}$, $W_{\delta=0.05}$, $W_{\delta=0.1}$, $W_{\delta=0.2}$, $W_{\delta=0.3}$ (solid lines, ordered top to bottom) as a function of $q_1$, with $p_1=0.03,\pi=0.0224$, $\pi_{11}=0.10$, $\pi_{12}=0.06$, $\pi_{22}=0.02$. Right plot: Power functions of $\T$ (dotted line) and $W_{\delta=0.01}$, $W_{\delta=0.05}$, $W_{\delta=0.1}$, $W_{\delta=\pi}$, $W_{\delta=0.2}$, $W_{\delta=0.3}$ (solid lines, ordered top to bottom) as a function of $q_1$, with $p_1=0.2, \pi=0.16$, $\pi_{11}=0.40, \pi_{12}=0.10, \pi_{22}=0.25$. }
\label{fig.power_robust}
\end{figure}

\subsection{Simulation study: Type I error and  power for finite samples}
\label{sec.Simulation_finite_samples}

By their design, the p-values of the considered tests are derived using the asymptotic normality of the underlying test statistics $\W$ and $\T$. In this section we study the finite sample behavior of $\W$, including its robustness against departures from HWE. In an applied setting, while other factors such as the MAF also play a role, it is the sample size that is the primary driving factor of the accuracy of the asymptotic normal approximation underlying the p-values of the tests. The primary goal of the simulation study is to investigate the type I error behavior in a finite sample setting for a variety of MAFs ranging between 0.03 and 0.5, which was the range of MAFs observed in the major depression disorder data set analyzed in Section \ref{sec.Application}. In a typical GWAS a whole range of MAFs are present, which means that an appropriate measure of the expected type I error, and the one used in our simulation study, is the weighted average of the observed type I errors over the entire range of MAFs with weights equal to the expected relative representation of each MAF in the study.

For each MAF we simulated the marker alleles for $R$ cases and for $S$ controls with $\Delta=0$ (i.e. under the null hypothesis of no association). In terms of the ratio of cases to controls we considered two scenarios, namely the balanced design with equal numbers of cases and controls ($R=S$) and an unbalanced design with the number of controls twice the number of cases ($S=2R$). We chose to focus on a setting with an excess of controls as it is typically easier to find individuals from the control population. The selected parameter values can be seen in Table \ref{tab.ResultsHWE}.

Under each parameter setting we simulated 5 billion data sets and for each of them we calculated the statistics $T$, $W_{\delta=0.05}$, $W_{\delta=0.1}$, $W_{\delta=0.2}$, $W_{\delta=0.3}$ and $W_{\delta=0.4}$. The tests were performed using the asymptotic standard normal approximation at the significance level $\alpha=5\times 10^{-8}$, a value that is typically used in GWAS. The observed type I error for each statistic and each selected MAF was calculated. Note that the number of simulated data sets (billions) had to be very high given the low level of significance, which in turn had to be set low in order to emulate a GWAS setting.


The overall type I error estimate was computed as a weighted average of the (estimated) type I errors at a dense grid of MAF values with weights based on the expected relative frequencies of each MAF. Given that in the real-life data set analyzed in Section \ref{sec.Application} the observed distribution of the MAFs was very close to uniform between 0.03 and 0.5, it was therefore deemed sufficient to calculate the overall type I error as a simple (i.e. unweighted) mean of the individual simulated error rates at the grid covering the interval from 0.03 to 0.5 (with steps of 0.005).

\subsubsection*{Type I error for finite samples}

The results of the simulation studies for the type I error are presented in Table \ref{tab.ResultsHWE}. It shows the ratios of the observed type I errors and the significance level $\alpha$ for fixed values of MAF (in all but the last column of the table). Given the small value of $\alpha$, we are in fact verifying the accuracy of the far-tail asymptotic normal approximation of the true distributions of the test statistics. The estimates of the expected overall type I error in a GWAS are given in the last column of the table.

The simulation results show that for $\T$ the observed type I error is slightly inflated for the unbalanced design, while for the balanced designs the statistic performs quite well. The table also shows that the statistic $\W$ exhibits a slightly inflated overall type I error, although it is worth noting that the inflation is considerably stronger for markers with MAFs below 0.05 and small values of $\delta$ and it steadily decreases with increasing sample size. This behavior appears to be a consequence of the low accuracy of the normal asymptotic approximation in the far tails of the distribution. Seeing that the results show a decreasing trend of the inflation with sample size, a remedy would be an increase of the underlying sample size. Crucially, despite the sub-optimal behavior of $\W$ for the very small MAFs, it needs to be stressed that the power gains achieved by $\W$ relative to the commonly used tests are not solely or even primarily due to the inflated type I error, since with growing sample size the type I error inflation vanishes while the superior power performance remains. With $R=S=50,000$ the inflation for $\delta=0.05$ is essentially gone.

Besides the increased sample size, an alternative remedy of the type I error inflation is to use a larger value for $\delta$ when calculating $\Wd$. In other words, aim to "overestimate" the population prevalence of cases if $\pi$ is small. This can be an especially effective solution if used only for markers with low MAF (e.g.\ below 0.1). Unsurprisingly of course, this "overestimation" approach does come at a price in terms of decreased power. A further alternative option is to obtain the p-values for $\W$ using a permutation approach. This can be done either for all markers or only for the markers for which the type I error is expected to be inflated (typically those with small MAF).

\begin{table}
\begin{center}
\caption{Type I errors divided by $\alpha=5\times 10^{-8}$. The column \textsl{total} shows the weighted average of type I errors over the various MAFs.}
\begin{tabular}{l||r|r|r|r|r|r|r|r}
   &  \multicolumn{8}{c}{$q_1$}\\
 & \multicolumn{1}{c|}{$0.03$}
 & \multicolumn{1}{c|}{$0.05$}
 & \multicolumn{1}{c|}{$0.10$}
 & \multicolumn{1}{c|}{$0.20$}
 & \multicolumn{1}{c|}{$0.30$}
 & \multicolumn{1}{c|}{$0.40$}
 & \multicolumn{1}{c|}{$0.50$}
 & \multicolumn{1}{c}{total} \\
\hline
\multicolumn{9}{l}{{\underline{$R=S=5000$}}}\\
$\T$ & 0.97 & 1.02 & 1.04 & 1.03 & 1.06 & 1.02  & 0.93 & 1.02 \\
$\Wde{0.05}$  & 6.95 & 3.88  & 2.08 & 1.35 & 1.16  & 1.03 & 0.93 & 1.62\\
$\Wde{0.10}$  & 5.41 & 3.17 & 1.84 &  1.24 & 1.16 & 1.00  & 0.93& 1.47\\
$\Wde{0.20}$  & 3.01 & 2.21 & 1.46 & 1.13 & 1.12 & 1.00 & 0.93 & 1.25\\
$\Wde{0.30}$  & 1.87 & 1.40 & 1.21& 1.07 & 1.08 & 0.99  & 0.93& 1.10\\
$\Wde{0.40}$  & 1.16 & 1.11 & 1.03 & 1.02 &  1.06 & 0.97   & 0.93& 1.02\\
\hline
\multicolumn{9}{l}{{\underline{$2R=S=6000$}}}\\
$\T$ & 2.25 & 1.84 & 1.32 & 1.10 & 1.06 & 0.94 & 0.95 & 1.16\\
$\Wde{0.05}$  & 5.83 & 3.28 & 1.79 & 1.22 & 1.14 & 0.98 & 0.91 & 1.46\\
$\Wde{0.10}$  & 4.24 & 2.58 & 1.60 & 1.16 & 1.07 & 0.98 & 0.90 & 1.32\\
$\Wde{0.20}$  & 2.27 & 1.61 & 1.23 & 1.05 & 0.98 & 0.96 & 0.89 & 1.10\\
$\Wde{0.30}$  & 1.25 & 1.12 & 1.02 & 1.03 & 0.98 & 0.92 & 0.89 & 0.99\\
$\Wde{0.40}$  & 0.84 & 0.94 & 0.95 & 0.98 & 0.97& 0.93 &  0.89 & 0.95\\
\hline
\multicolumn{9}{l}{{\underline{$R=S=10,000$}}}\\
$\T$ & 0.91 & 0.93 & 1.00 & 1.14 & 0.92& 1.15 & 1.10 & 1.05\\
$\Wde{0.05}$  & 3.73 &2.31 & 1.58 & 1.27& 0.97& 1.11 & 1.09 & 1.33\\
$\Wde{0.10}$  & 3.02 & 1.99 & 1.43& 1.20& 0.96&  1.12 & 1.08 & 1.25\\
$\Wde{0.20}$  & 2.05 & 1.51 & 1.22 & 1.18& 0.93 & 1.12 & 1.06 & 1.15\\
$\Wde{0.30}$  & 1.44 & 1.19 & 1.06 & 1.15 & 0.91 & 1.11& 1.05 & 1.07\\
$\Wde{0.40}$  & 1.07 & 1.02 & 0.98 & 1.12 & 0.92 & 1.12& 1.05& 1.04\\
\hline
\multicolumn{9}{l}{{\underline{$R=S=20,000$}}}\\
$\T$ & 1.02 & 1.05 & 0.92 & 1.02 & 1.06 & 1.12 & 1.06 & 1.04 \\
$\Wde{0.05}$  & 2.04 & 1.74 & 1.22 & 1.07 & 1.10 & 1.09 & 1.06 & 1.17\\
$\Wde{0.10}$  & 1.79 & 1.58 & 1.14 & 1.05 & 1.09 & 1.09 & 1.06 & 1.14\\
$\Wde{0.20}$  & 1.37 & 1.34 & 1.03 & 1.04 & 1.08 &1.08 & 1.06& 1.09\\
$\Wde{0.30}$  & 1.10 & 1.16 & 0.99 & 1.02 & 1.08 & 1.09& 1.06& 1.06\\
$\Wde{0.40}$  & 1.02 & 1.06 & 0.95 & 1.02 & 1.05 & 1.10& 1.06 & 1.04
\end{tabular}
\end{center}
\label{tab.ResultsHWE}
\end{table}

\subsubsection*{Power (type II error) for finite samples}

Besides the type I error investigation, we also compared the power performances of the various tests in a finite sample size setting. We simulated data under a number of parameter combinations replicating each test 5 million times under each scenario. The significance level was again set at $5\times10^{-8}$ and the empirical power of each test was calculated as the fraction of time the test rejected the null at this level of significance separately for each scenario. The analysis showed that the finite sample empirical power functions are very similar to the asymptotic power functions (plots not shown).

\subsection{Adjustments under Hardy-Weinberg Disequilibrium (HWD)}
\label{sec.AdjustmentsHWD}

Many existing tests, including ours so far, implicitely rely on the validity of the assumption of Hardy-Weinberg equilibrium (HWE). In applications where such assumption is expected not to be appropriate, the usual approach is to rely on the Cochran-Armitage trend test (CATT) and its robustness against departures from HWE. Advantageously, our newly proposed test statistic $\W$ as well as $\T$ can both be robustified against departures from HWE.

In \cite{Zheng12,Schaid99} an adjusted test statistic $\THWD$ is described. It is found by replacing the estimated products $q_{1}^Uq_{2}^U$ and $q_{1}^Aq_{2}^A$ in the denominators of $\T$ by suitable estimators of $q_{1}^Uq_{2}^U+q_{11}^U-(q_{1}^U)^2$ and $q_{1}^A q_{2}^A+q_{11}^A-(q_{1}^A)^2$, where $q_{11}^U$ and $q_{11}^A$ denote the frequency of the genotype $(M_1,M_1)$ among the controls and among the cases, respectively. The adjustment follows from expressions for the variances of $\hqone^U$ and $\hqone^A$ derived without the assumption of HWE. The unknown frequencies in these expressions are estimated using the corresponding sample frequencies. The asymptotic normality of the (adjusted) test statistic $\THWD$ under the null hypothesis follows by the central limit theorem. Conveniently, the test statistic $\W$ can be adjusted in an analogous way. To that end we define
\begin{align*}
\WHWD\;=\;\frac{\sqrt{m}(\hqone^U-\hqone^A)}
{\sqrt{\hqpione\hqpitwo+(\hqhpioneone-\hqpione^2)}},
\end{align*}
where $\hqhpioneone$ is an estimate of $q_{11}$ defined analogously to $\hqpione$ of (\ref{eqn.hqhpii}).

\medskip

We also performed simulations that compared the performance of the adjusted test statistic $\WHWD$ with the test $\Zcatt{}(1/2)$ for different values of $\delta$, to investigate the combination of both the robustness of the test in case of a misspecified population prevalence and deviation from HWE. The results (Appendix B) show that in the considered settings the test based on $\WHWD$ is slightly more powerful than the CATT.

\subsection{Multiple causal variants}
\label{sec.Multiple_variants}

Until this point we focused on testing for association between markers and a single causal variant. We showed that the p-values of $\W$ can be used for identification of markers in strong LD with the causal variant in a way that does not preferentially select markers with high MAF, a property that is rooted in the equation (\ref{eqn.approx_W_and_T}), where the term $B$ on the right-hand side is the same for all considered markers. Unfortunately, the argument only applies to the situations with a single variant, given that the term $B$ is causal-variant specific. Since the power function and the p-value of $\W$ strongly depends on the value of $B$, only the p-values of markers that are in LD with the same causal variant can be directly used as measures of the degree of LD with one of the causal variants. In practice, this means that one needs to be careful when comparing p-values of markers that are located far apart on the genome, especially if they are located on different chromosomes. This holds for all tests mentioned in this paper. For markers that are in LD with multiple causal variants it is in general very complicated to quantify the corresponding effect towards the p-values of all of these tests.

\section{Application}
\label{sec.Application}

The newly proposed test based on $\W$ was applied to a case-control data set to identify the genomic regions that confer risk for and protection against major depressive disorder (MDD). To this day, the efforts to identify such regions have not been very successful \cite{Boomsma08,Wray18}. While this might be partially due to a lack of consensus on the exact definition of the condition (MDD) itself, the possibility that the disease is influenced by many genetic loci each with a small marginal effect could be an even better explanation for the lack of success. Then, in order to detect these loci, a more sensitive statistical test appears to be needed. We believe that our novel test statistic might be able to at least partially answer that call.

The analyzed MDD data set was primarily collected using two databases collected in the Netherlands \cite{Boomsma08}. The cases, i.e. the individuals affected by MDD, came primarily from the purpose-specific Netherlands Study of Depression and Anxiety (NESDA) database, while the controls came from the Netherlands Twin Registry (NTR), a database containing primarily data about twin siblings and their parents. In order to achieve sufficient independence among the controls, for each pedigree from NTR a single individual was randomly selected, thereby making all individuals in the data set (biologically) unrelated and thus statistically independent. The total number of cases and controls equaled 2306 and 1027, respectively. In the analysis, only markers for which the sum of the minor allele frequency among the cases and the controls was at least 0.02 were included, resulting in a data set with over 600K markers. The population prevalence of MDD in the relevant population is accurately estimated as $\hat\pi=0.15$ \cite{Boomsma08}, which is the value we used in the test statistic $\W$.

For all markers in the database the p-values corresponding to $\W$ and $\T$ were computed using their asymptotic distribution. Figure \ref{fig.p-value_scatter} shows a scatter-plot of the p-values for the two test statistics. The markers were divided into three categories based on their MAFs in the sample of controls. The three plots show the degree of dissimilarity of the two statistics. The observed general pattern is that for high MAF the p-values of the two tests are highly similar and with decreasing MAF they become increasingly dissimilar.

\begin{figure}
\begin{center}
\includegraphics[scale=0.45]{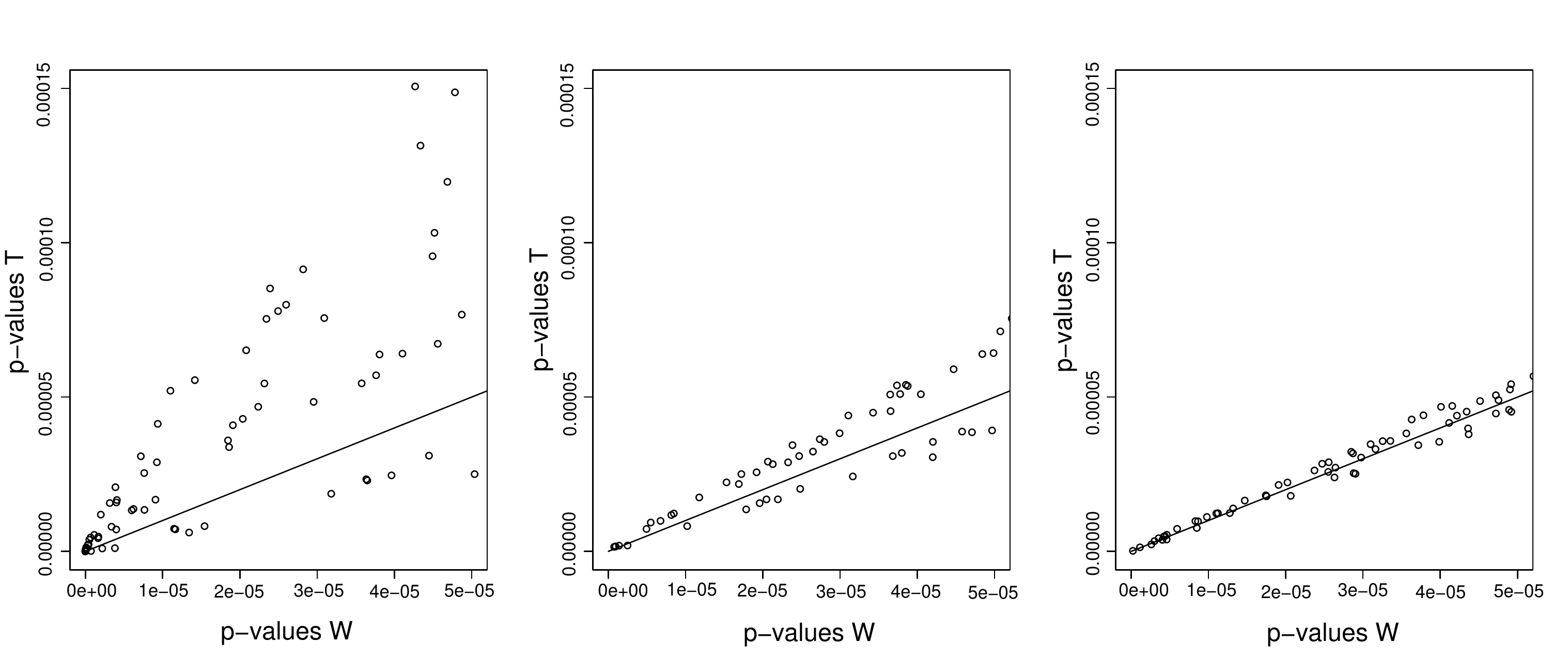}
\end{center}
\caption{Scatter-plots of the p-values for the test statistics $\T$ against those based on $W_{0.15}$. From left to right: $\hqone^U\leq0.10,\hqone^U\in(0.10,0.25],\hqone^U\in(0.25,0.50]$.}
\label{fig.p-value_scatter}
\end{figure}

The analysis yielded several markers with p-values below the threshold $\alpha=5.0 \times 10^{-8}$. With the exception of chromosome 15, all other chromosomes had no more than a single significant marker. On chromosome 15 we identified 6 markers with p-values below $\alpha$. These were the markers with RS numbers rs10152733 (MAFs 0.0732 and 0.0245 among the cases and controls, respectively), rs3784362 (0.0212/0.00734), rs1463912 (0.0249/0.00831), rs7168666 (0.0216/0.00734), rs1820416 (0.0210/0.00734), rs4777166 (0.0202/0.00642). Moreover, two of these markers, namely rs10152733 and rs1463912, were also identified as significant by the statistic $\T$. Figure 6 shows a Manhattan plot of the $-\log(p$-value) based on $W_{\hat\pi=0.15}$ for chromosome 15.

\begin{figure}
\begin{center}
\includegraphics[scale=0.45]{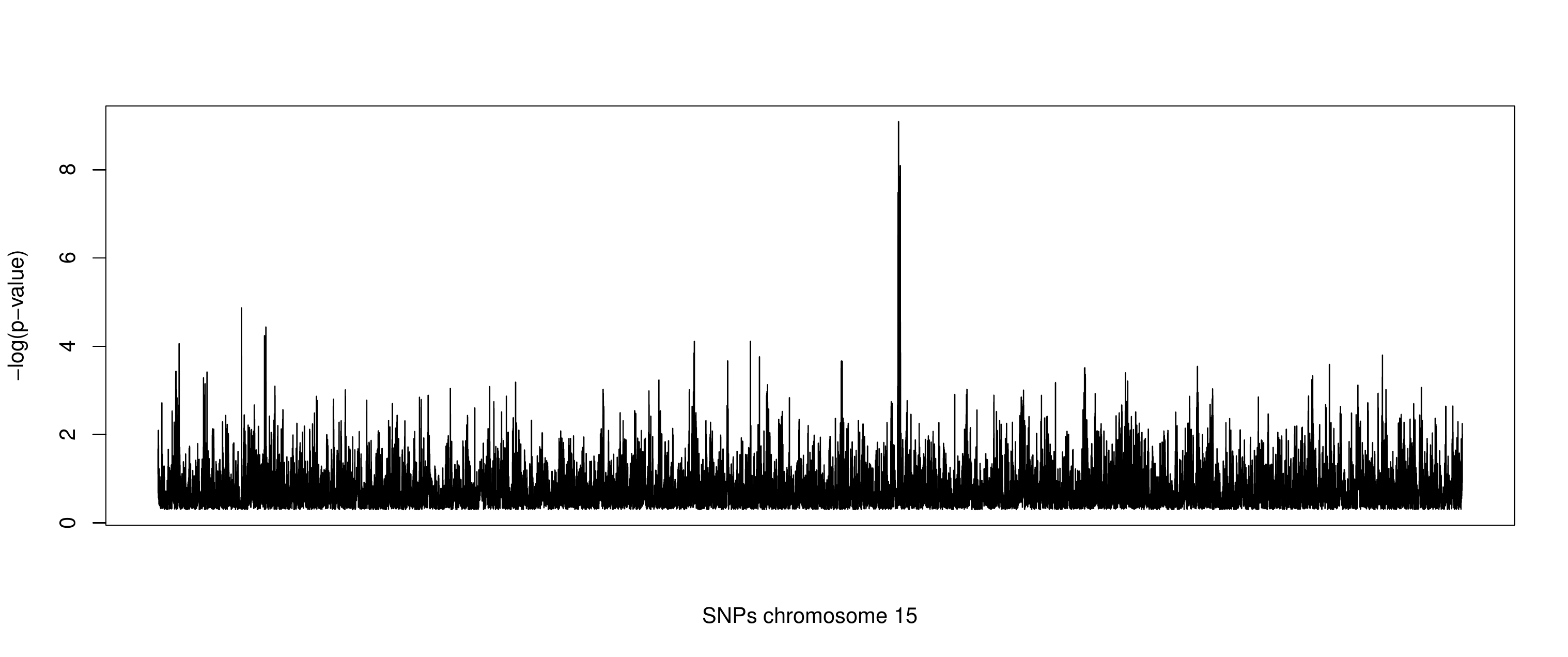}
\end{center}
\caption{Manhattan plot of the $-\log(p$-value) based on $W_{\hat\pi=0.15}$ for chromosome 15.}
\label{fig.manhattan}
\end{figure}

\cite{Boomsma08} contains the results of 10 GWAS studies for MDD. The data for these studies were collected in relation to various conditions, not solely MDD. In addition to MDD, these included recurrent MDD, alcoholism or nicotine dependence, and others. Some of the studies showed significant association between genomic regions and MDD, however, none of these were replicated by an independent study. Among the markers identified by our analysis none have been found by one of the studies in \cite{Boomsma08}. In \cite{Wray18} a genome-wide association meta-analysis based in 135,458 cases and 344,901 controls
was conducted. They identified 44 loci. None of these were replicated by our analysis.

\section{Discussion}

Various strategies for selecting markers for follow-up studies using the case-control framework have been proposed in the literature. As discussed in this paper, these include the allele-based test of association which assesses the difference of marker allele frequencies between the cases and the controls, the LRM score test, the chi-square test for association, and the Cochran-Armitage trend test. An often observed shortcoming of the existing strategies is their preference for markers with high MAFs at the expense of markers with low MAFs.

In this paper a novel allele-frequency-based test statistic for finding association between genetic markers and a disease of interest is proposed. A competitive advantage of the statistic is that it does not favor markers with high MAFs. In light of the known and suspected importance of rare alleles, this means that our new test is much more suitable to be used to asses the association of genetic markers with a disease based on the observed p-values.

An additional advantage of the new test is that the statistic can be efficiently computed using basic summary statistics of the case-control sample. We derived the asymptotic power function of the test, which allows for efficient computation of the associated p-values, an important strength especially compared to approaches that rely on permutation schemes in order to obtain the p-values.

We studied the power performance of the newly proposed test and compared it with a number of commonly used alternative tests under numerous scenarios. The obtained results were favorable for the new test. It was observed that compared to the existing tests the new test possesses superior power for markers with low MAF. This behavior is unsurprising in light of the fact that the power functions of the new tests are (nearly) constant for various marker allele frequencies, while the power functions of the competing tests generally decline with decreasing MAF.

The calculation of the newly proposed test statistic requires the estimation of the prevalence of the disease $\pi$. This value cannot be directly obtained from the sampled case-control data alone and the estimation requires external data. This, however, is not a major obstacle for the usage of the test since for many diseases suitable estimates of population prevalence are readily available from sources such as national registries. Furthermore, we showed that the novel statistic is fairly robust against misspecification of the prevalence parameter, which means that even when an accurate estimate of the prevalence is not available for the population of interest, an inaccurate (over-)estimate (e.g. based on a related population) can be used without substantially harming the power of the resulting test.

Besides power, we also studied the type I error of the new test in the context of a finite sample setting. The simulations give evidence that the type I error of the new test is inflated for small MAFs and low prevalence $\pi$. The specific degree of inflation depends on the underlying sample size. However, we observed that the inflation decreases with increasing sample size and therefore cannot be the reason for the observed power gains of the new test. Moreover, our simulation showed that the overall type I error for the new test is expected to be only slightly inflated in the context of a genome-wide study with a broad range of MAF values.

For many traits, only a small proportion of the variability in the population can be explained by causal variants that have been identified so far \cite{Manolio09}. One possible explanation for this "missing heritability" is the presence of low-frequency variants with relatively strong effect on disease risk. Indeed, rare variants found by re-sequencing have already been described to affect complex diseases \cite{Schork09}. Given the properties of our test statistics, and in the light of the current interest in detecting association between complex phenotypes and low-frequency variants and locating causal variants with small minor allele frequencies \cite{Sham14}, we believe that the novel method presented in this paper could prove to be a very useful addition to the landscape of methods available for tackling these important problems of genetics.

\section*{Acknowledgement}

We express gratitude to the Netherlands Twin Registry (NTR) and de Nederlandse Studie naar Depressie en Angst (NESDA) for making the data available for application of the theory in this paper. The data was collected with support from NWO (904-61-090; 904-61-193; 480-04-004; 400-05-717; SPI 56-464-14192), Center for Medical Systems Biology (NWO Genomics); the EU (EU/QLRT-2001-01254); Geestkracht program of ZonMW (10-000-1002), Neuroscience Campus Amsterdam (NCA) and the EMGO+ institute; and institutes involved in NESDA (VU University Medical Centre, Leiden University Medical Centre, GGZinGeest, Rivierduinen, University Medical Centre Groningen, GGZ Lentis, GGZ Friesland, GGZ Drenthe); the Genetic Association Information Network (GAIN); ARRA grants 1RC2 MH089951-01 and 1RC2MH089995-01; FP7-HEALTH-F4-2007-201413; European Research Council (ERC-230374).

\section*{Conflict of interest statement}

The authors have declared no conflict of interest.

\nocite{*}
\bibliographystyle{unsrtnat}
\bibliography{references}

\section*{Appendix A: derivation of the equality (\ref{eqn.pMLUDelta})}

In this section we formulate two lemmas which together constitute the proof of equality (\ref{eqn.pMLUDelta}). We note that throughout Appendix A (and only there) we assume that the genotypes are ordered. For simplicity of notation, without loss of generality, we assume that the total number of causal variants equals two.

\begin{lemma}
Let there be two causal variants and a marker of interest. Denote the genotypes at the two causal variants as $(A_i,A_j)$ and $(B_i,B_j)$ with $i,j=1,2$, respectively. Suppose that the marker is in linkage disequilibrium with the first causal variant and in linkage equilibrium with the second causal variant. Then,
\begin{align*}
\rP(M_i M_j | A) &\;=\;q_{ij} + D_{11ij}\frac{\pi_{11}-\pi_{12}}{\pi} + D_{22ij}\frac{\pi_{22}-\pi_{12}}{\pi}, \\
\rP(M_i M_j | U) &\;=\;q_{ij} + D_{11ij}\frac{\pi_{12}-\pi_{11}}{1-\pi} + D_{22ij}\frac{\pi_{12}-\pi_{22}}{1-\pi},
\end{align*}
where $D_{ijkl}=\rP(A_iA_jM_kM_l)-\rP(A_iA_j)\rP(M_kM_l)$, and $\rP(A_iA_jM_kM_l)$ equals the probability that a random individual from the total population has haplotypes $(A_i,M_k)$ and $(A_j,M_l)$. Moreover, if HWE holds, then also $D_{ijkl}=\rP(A_iM_k)\rP(A_jM_l)-p_ip_jq_kq_l$.
\end{lemma}

\begin{proof}
Write
\begin{align*}
\rP(M_iM_j | A)&\;=\;\sum_{k,l,m,n} \rP(M_iM_j | A_k A_l B_m B_n) \rP(A_k A_l B_m B_n | A)\\
&\;=\; \sum_{k,l,m,n} \rP(M_iM_j | A_k A_l) \rP(A_k A_l B_m B_n | A)\\
&\;=\; \sum_{k,l} \rP(M_iM_j | A_k A_l) \rP(A_k A_l | A)\\
&\;=\; q_{ij} + \rP(M_iM_j | A_1 A_1)\frac{p_{11}\pi_{11}}{\pi} +
(p_{12}\rP(M_iM_j | A_1 A_2)+p_{21}\rP(M_iM_j | A_2 A_1))\frac{\pi_{12}}{\pi} \\
&\qquad + \; \rP(M_iM_j|A_2 A_2)\frac{p_{22}\pi_{22}}{\pi} - \frac{q_{ij} (p_{11}\pi_{11} + (p_{12}+p_{21})\pi_{12} + p_{22}\pi_{22})}{\pi} \\
&\;=\; q_{ij} + \frac{D_{11ij}\pi_{11}}{\pi} + \frac{(D_{12ij}+D_{21ij})\pi_{12}}{\pi}+\frac{D_{22ij}\pi_{22}}{\pi} \\
&\;=\; q_{ij} + D_{11ij}\frac{\pi_{11}-\pi_{12}}{\pi} + D_{22ij}\frac{\pi_{22}-\pi_{12}}{\pi},
\end{align*}
where the last equality follows from $D_{12ij}+D_{21ij}=-(D_{11ij} + D_{22ij})$, since $\sum_{k,l=1}^{2,2}D_{klij}=0$. The expression for $\rP(M_iM_j | U)$ is found analogously. The assertion requiring HWE is trivial.
\end{proof}

\begin{lemma}
\label{lem.freqs}
The frequencies of the allele $M_k$ among the cases and the controls satisfy
\begin{align}
\label{eqn.Mk}
q_{k}^A\;=\;\rP(M_k| A)
&\;=\;q_{k} + p_1 D_{1k}\frac{\pi_{11}-\pi_{12}}{\pi} - p_2 D_{1k}\frac{\pi_{22}-\pi_{12}}{\pi} \\
&\;=\; q_{k} + \sqrt{p_1p_2q_1q_2} \Delta_{1k} \Big(p_1 \frac{\pi_{11}-\pi_{12}}{\pi} - p_2 \frac{\pi_{22}-\pi_{12}}{\pi}\Big) \nonumber\\
q_{k}^U\,=\,\rP(M_k| U)
&\;=\; q_k +  p_1 D_{1k} \frac{\pi_{12}-\pi_{11}}{1-\pi} - p_2 D_{1k} \frac{\pi_{12}-\pi_{22}}{1-\pi}\nonumber \\
&\;=\;q_k +  \sqrt{p_1p_2q_1q_2} \Delta_{1k} \Big(p_1\frac{\pi_{12}-\pi_{11}}{1-\pi} - p_2 \frac{\pi_{12}-\pi_{22}}{1-\pi}\Big),\nonumber
\end{align}
where $\Delta_{ik}=D_{ik}/\sqrt{p_1p_2q_1q_2}$, $D_{ik} = \rP(A_iM_k) - p_iq_k$ for $i,k=1,2$, with $\rP(A_iM_k)$ denoting the $(A_i,M_k)$-haplotype frequency in the total population.
\end{lemma}

\begin{proof}
Define $\bar k=3-k$. Then
\begin{align*}
q_{k}^A&\;=\;\frac{1}{2}\rP(M_kM_{\bar k}| A) + \frac{1}{2}\rP(M_{\bar k} M_k| A)+ \rP(M_kM_k| A) \\
&\;=\;q_{k} + \Big(\frac{1}{2}D_{11k\bar k}+\frac{1}{2}D_{11\bar k k}+D_{11kk}\Big)\frac{\pi_{11}-\pi_{12}}{\pi}
 + \Big(\frac{1}{2}D_{22k\bar k}+\frac{1}{2}D_{22\bar k k}+D_{22kk}\Big)\frac{\pi_{22}-\pi_{12}}{\pi}\\
&\;=\;q_{k} + p_1 D_{1k}\frac{\pi_{11}-\pi_{12}}{\pi} + p_2 D_{2k}\frac{\pi_{22}-\pi_{12}}{\pi}.
\end{align*}
Since $D_{1k}=-D_{2k}$, the above expression further equals the right-hand side of (\ref{eqn.Mk}). The expression for $q_{k}^U$ is found analogously.
\end{proof}

A crucial consequence of Lemma \ref{lem.freqs} is the equality
\begin{align*}
q_{1}^U - q_{1}^A \; = \; \sqrt{p_1p_2q_1q_2}\Delta_{11} \; \frac{p_1(\pi_{12}-\pi_{11}) + p_2 (\pi_{22}-\pi_{12})}{\pi(1-\pi)}.
\end{align*}
In the case that the marker is in fact the causal variant (i.e. $\Delta_{11}=1, p_1=q_1$ and $p_2=q_2$), we get
\begin{align*}
p_{1}^U - p_{1}^A \;=\; p_1p_2 \; \frac{ p_1(\pi_{12}-\pi_{11}) + p_2(\pi_{22}-\pi_{12}) }{\pi(1-\pi)}.
\end{align*}
Combining the two displays immediately yields (\ref{eqn.pMLUDelta})

\section{Appendix B: Comparison of $\WHWD$ and $\Zcatt{}(1/2)$}

First we focus on the type I error behavior. We simulated data under the null hypothesis of no association with 2000 cases and 2000 controls. The MAF of the marker was again varied over a broad range of values and the Wright's inbreeding coefficient was alternatively set to $0.1$ and $0.2$. Using the significance threshold of $\alpha=0.05$, we repeated one million times the simulation of data and hypothesis testing. The observed type I error rates were all close to $\alpha$, like it should be (results not shown).

Next, we performed simulations to compare the power of the $\Zcatt{}(1/2)$ and $\WHWD$ statistics  under deviations of Hardy Weinberg equilibrium and misspecified population prevalence. The allele frequency $p_1$ was set equal to 0.03 and the Wright's inbreeding coefficient for the causal variant was alternatively set to $0.1$ and $0.2$. Given the non-zero value of $\Delta$ ($\Delta = 0.10$), the alleles at the marker were simulated to also be in HWD. We set $R=S=4000$ and $\pi_{11}=0.10$, $\pi_{12}=0.06$ and $\pi_{22}=0.02$. The observed power functions are plotted in Figure \ref{fig.ResultsHWD}, which shows a slight power superiority of the test based on $\WHWD$ over the Cochran-Armitage test. Furthermore, the power of the test $\WHWD$ as a function of the MAF $q_1$ is constant, once again illustrating how even the robustified version of the new test does not unjustly prefer markers with high MAF.

\begin{figure}
\begin{center}
\includegraphics[scale=0.55]{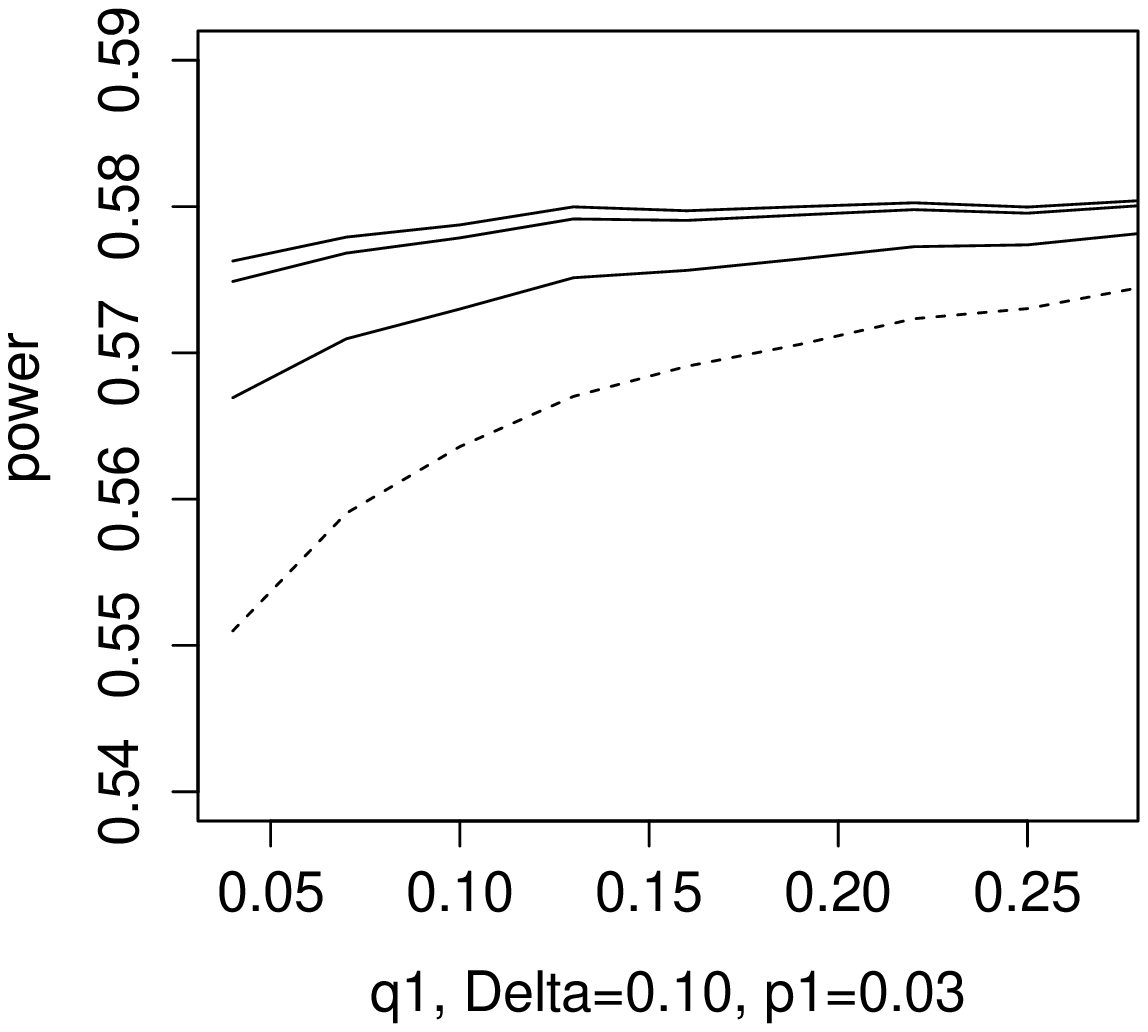}
\includegraphics[scale=0.55]{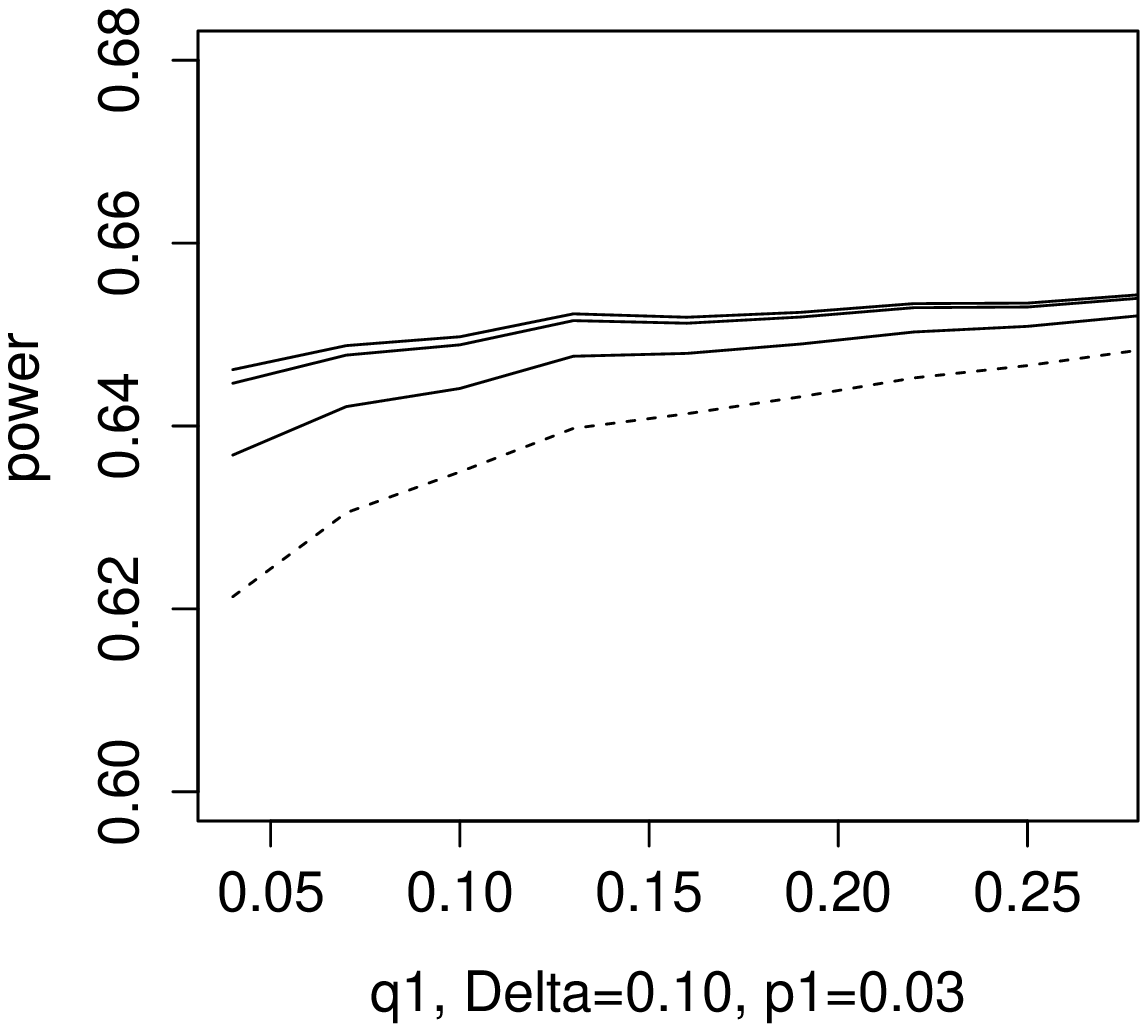}
\end{center}
\caption{Power functions for the test based on $\WHWD$ (continued lines, for $\delta=\pi,\delta=0.05,\delta=0.20$ ordered top to bottom) and $\Zcatt{}(1/2)$ (dashed lines), for $F$ equal to $0.1$ (left plot) and $0.2$ (right plot) as a function $q_1$.}
\label{fig.ResultsHWD}
\end{figure}

\end{document}